\documentclass[10pt, twocolumn, aps, nofootinbib, longbibliography, nobibnotes]{revtex4-1} 
\pdfoutput=1	

\usepackage{amsfonts, amssymb, amsmath, amsthm}
\usepackage{latexsym}
\usepackage[tracking=smallcaps]{microtype}	
\usepackage{url}
\usepackage{color}
\definecolor{DarkGray}{rgb}{0.1,0.1,0.5}
\usepackage[colorlinks=true,breaklinks, linkcolor=DarkGray,citecolor=DarkGray,urlcolor=DarkGray]{hyperref}	

\usepackage{graphicx}
\usepackage[tight, TABBOTCAP]{subfigure}

\newcommand{\bra}[1]{{\langle#1|}}
\newcommand{\ket}[1]{{|#1\rangle}}

\newcommand{\ketbra}[2]{{\ket{#1}\!\bra{#2}}}

\newcommand{\abs}[1]{{\lvert #1\rvert}}	

\newcommand{\identity}{\ensuremath{\boldsymbol{1}}} 

\def\CZ {C\!Z}	

\newcounter{sprows}

\newlength{\spheight}
\newlength{\spraise}

\newcommand{\comment}[1]{\emph{\color{blue}Comment:\color{black} #1}} 
\newlength{\commentslength}
\newcommand{\comments}[1]{
\hspace{-2\parindent}
\addtolength{\commentslength}{-\commentslength}
\addtolength{\commentslength}{\linewidth}
\addtolength{\commentslength}{-\parindent}
\fcolorbox{blue}{white}{\smallskip\begin{minipage}[c]{\commentslength}
\emph{Comments:}\begin{itemize}#1\end{itemize}\end{minipage}}\bigskip
}
\newcommand{\rem}[1]{}

\newtheorem{theorem}{Theorem}

\newtheorem{claim}[theorem]{Claim}

\newfont{\subsubsecfnt}{ptmri8t at 11pt}
\renewcommand{\subparagraph}[1]{\smallskip{\subsubsecfnt #1.}}

\newcommand{\eqnref}[1]{\hyperref[#1]{{(\ref*{#1})}}}
\newcommand{\thmref}[1]{\hyperref[#1]{{Theorem~\ref*{#1}}}}
\newcommand{\lemref}[1]{\hyperref[#1]{{Lemma~\ref*{#1}}}}
\newcommand{\corref}[1]{\hyperref[#1]{{Corollary~\ref*{#1}}}}
\newcommand{\defref}[1]{\hyperref[#1]{{Definition~\ref*{#1}}}}
\newcommand{\secref}[1]{\hyperref[#1]{{Sec.~\ref*{#1}}}}
\newcommand{\figref}[1]{\hyperref[#1]{{Fig.~\ref*{#1}}}}  
\newcommand{\tabref}[1]{\hyperref[#1]{{Table~\ref*{#1}}}}
\newcommand{\remref}[1]{\hyperref[#1]{{Remark~\ref*{#1}}}}
\newcommand{\appref}[1]{\hyperref[#1]{{Appendix~\ref*{#1}}}}
\newcommand{\claimref}[1]{\hyperref[#1]{{Claim~\ref*{#1}}}}
\newcommand{\factref}[1]{\hyperref[#1]{{Fact~\ref*{#1}}}}
\newcommand{\propref}[1]{\hyperref[#1]{{Proposition~\ref*{#1}}}}
\newcommand{\exampleref}[1]{\hyperref[#1]{{Example~\ref*{#1}}}}
\newcommand{\conjref}[1]{\hyperref[#1]{{Conjecture~\ref*{#1}}}}

\allowdisplaybreaks[1]

\def\COLOR{}
\ifdefined\COLOR

\else

\fi

\definecolor{Cayenne}{rgb}{0.5,0,0}
\definecolor{Midnight}{rgb}{0,0,0.5}
\definecolor{Plum}{rgb}{0.5,0,0.5}
\definecolor{Teal}{rgb}{0,0.5,0.5}
\definecolor{Clover}{rgb}{0,0.5,0}
\definecolor{Maroon}{rgb}{0.5,0,0.25}
\definecolor{Ocean}{rgb}{0,0.25,0.5}
\definecolor{Tangerine}{rgb}{1,0.5,0}
\definecolor{Strawberry}{rgb}{1,0,0.5}
\definecolor{Fern}{rgb}{0.25,0.5,0}
\definecolor{Aqua}{rgb}{0,0.5,1}
\definecolor{Moss}{rgb}{0,0.5,0.25}
\definecolor{Mocha}{rgb}{0.5,0.25,0}
\definecolor{Lemon}{rgb}{1,1,0}
\definecolor{Asparagus}{rgb}{0.5,0.5,0}
\definecolor{Grape}{rgb}{0.5,0,1}
\definecolor{Iron}{rgb}{.3,.3,.3}
\definecolor{Steel}{rgb}{.4,.4,.4}

\usepackage{enumitem} 

\makeatletter
\let\save@mathaccent\mathaccent
\newcommand*\if@single[3]{%
  \setbox0\hbox{${\mathaccent"0362{#1}}^H$}%
  \setbox2\hbox{${\mathaccent"0362{\kern0pt#1}}^H$}%
  \ifdim\ht0=\ht2 #3\else #2\fi
  }
\newcommand*\rel@kern[1]{\kern#1\dimexpr\macc@kerna}
\newcommand*\widebar[1]{\@ifnextchar^{{\wide@bar{#1}{0}}}{\wide@bar{#1}{1}}}
\newcommand*\wide@bar[2]{\if@single{#1}{\wide@bar@{#1}{#2}{1}}{\wide@bar@{#1}{#2}{2}}}
\newcommand*\wide@bar@[3]{%
  \begingroup
  \def\mathaccent##1##2{%
    \let\mathaccent\save@mathaccent
    \if#32 \let\macc@nucleus\first@char \fi
    \setbox\z@\hbox{$\macc@style{\macc@nucleus}_{}$}%
    \setbox\tw@\hbox{$\macc@style{\macc@nucleus}{}_{}$}%
    \dimen@\wd\tw@
    \advance\dimen@-\wd\z@
    \divide\dimen@ 3
    \@tempdima\wd\tw@
    \advance\@tempdima-\scriptspace
    \divide\@tempdima 10
    \advance\dimen@-\@tempdima
    \ifdim\dimen@>\z@ \dimen@0pt\fi
    \rel@kern{0.6}\kern-\dimen@
    \if#31
      \overline{\rel@kern{-0.6}\kern\dimen@\macc@nucleus\rel@kern{0.4}\kern\dimen@}%
      \advance\dimen@0.4\dimexpr\macc@kerna
      \let\final@kern#2%
      \ifdim\dimen@<\z@ \let\final@kern1\fi
      \if\final@kern1 \kern-\dimen@\fi
    \else
      \overline{\rel@kern{-0.6}\kern\dimen@#1}%
    \fi
  }%
  \macc@depth\@ne
  \let\math@bgroup\@empty \let\math@egroup\macc@set@skewchar
  \mathsurround\z@ \frozen@everymath{\mathgroup\macc@group\relax}%
  \macc@set@skewchar\relax
  \let\mathaccentV\macc@nested@a
  \if#31
    \macc@nested@a\relax111{#1}%
  \else
    \def\gobble@till@marker##1\endmarker{}%
    \futurelet\first@char\gobble@till@marker#1\endmarker
    \ifcat\noexpand\first@char A\else
      \def\first@char{}%
    \fi
    \macc@nested@a\relax111{\first@char}%
  \fi
  \endgroup
}
\makeatother

\def\llbracket{{[\![}}
\def\rrbracket{{]\!]}}
\DeclareMathOperator{\GL}{\operatorname{GL}}
\renewcommand{\CZ}{{\ensuremath{\mathrm{C}Z}}}
\newcommand{\CCZ}{{\ensuremath{\mathrm{CC}Z}}}

\def\beq{\begin{equation}}
\def\eeq{\end{equation}}

\renewcommand{\comment}[1]{}
\renewcommand{\comments}[1]{}

\begin{document}
\def\compilefullpaper{}

\title{Fault-tolerant quantum computation with few qubits}
\author{Rui Chao}
\author{Ben W. Reichardt}
\affiliation{University of Southern California}

\begin{abstract}
Reliable qubits are difficult to engineer, but standard fault-tolerance schemes use seven or more physical qubits to encode each logical qubit, with still more qubits required for error correction.  The large overhead makes it hard to experiment with fault-tolerance schemes with multiple encoded~qubits.  

The $15$-qubit Hamming code protects seven encoded qubits to distance three.  We give fault-tolerant procedures for applying arbitrary Clifford operations on these encoded qubits, using only two extra qubits, $17$ total.  In particular, individual encoded qubits within the code block can be targeted.  Fault-tolerant universal computation is possible with four extra qubits, $19$ total.  The procedures could enable testing more sophisticated protected circuits in small-scale quantum devices.  

Our main technique is to use gadgets to protect gates against correlated faults.  We also take advantage of special code symmetries, and use pieceable fault tolerance.  
\end{abstract}

\maketitle

\section{Introduction}

Quantum computers are faulty, but schemes to tolerate errors incur a large space overhead.  For example, one qubit encodes into seven physical qubits using the Steane code~\cite{Steane96css}, or into nine physical qubits using the Bacon-Shor and smallest surface codes~\cite{Bacon05operator, BombinMartindelgado07surfaceoverhead, HorsmanFowlerDevittVanmeter11surfacesurgery}.  Error correction uses additional qubits.  When more than one level of encoding is required for better protection, the overhead multiplies, so that thousands of physical qubits can be required for each logical qubit~\cite{Suchara13overhead}.  This overhead will compound the challenge of building large quantum computers.  In the near term, it also makes it more difficult to run fault-tolerance experiments, which are important to test different schemes' performance, validate models and learn better approaches.  

Codes storing multiple qubits have higher rates~\cite{CalderbankShor96}, but too large codes tend to tolerate less noise since initializing codewords gets difficult~\cite{Steane03}.  A key obstacle for using any code with multiple qubits per code block is that it is complicated and inefficient to address the individual encoded qubits to compute on them~\cite{Gottesman97, SteaneIbinson03multiqubitcodes}.  For example, to apply a CNOT gate between two logical qubits in a code block, the optimized method in~\cite{SteaneIbinson03multiqubitcodes} requires a full ancillary code block, with no stored data, into which the target logical qubit is transferred temporarily.  

We introduce lower-overhead methods for computing fault tolerantly on multiple data qubits in codes of distance two or three.  

\begin{enumerate}[leftmargin=*]
\item 
For even $n$, the $\llbracket n, n-2, 2 \rrbracket$ code encodes $n-2$ logical qubits into $n$ physical qubits, protected to distance two.  
We show that with two more qubits, encoded CNOT and Hadamard gates can be applied fault tolerantly.  For $n \geq 6$, four extra qubits suffice to fault-tolerantly apply an encoded $\CCZ$ gate, for universality.  
\item 
For better, distance-three protection, we encode seven qubits into $15$, and give fault-tolerant circuits for the encoded Clifford group using two more qubits, and for a universal gate set with four extra qubits ($19$ total).  
\end{enumerate}
Combined with the two-qubit fault-tolerant error-detection and error-correction methods in~\cite{ChaoReichardt17errorcorrection}, this means that substantial quantum calculations can be implemented fault tolerantly in a quantum device with fewer than $20$ physical qubits.  Figure~\ref{f:computationqubits} summarizes our results.  

\begin{figure}[b]
\centering
\includegraphics[scale=.455]{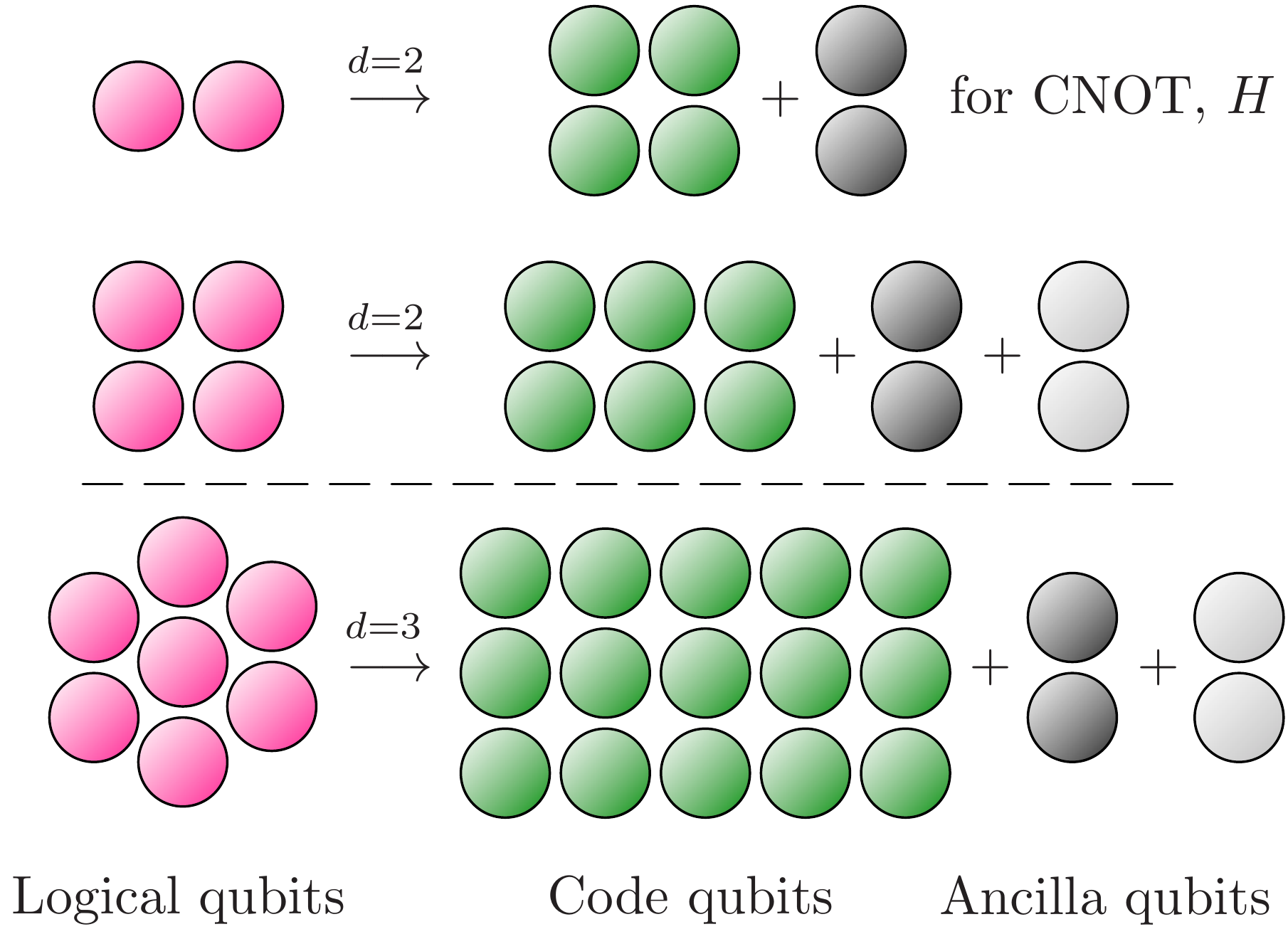}
\caption{Summary of our constructions.  Using two extra qubits, one can apply fault tolerantly either encoded CNOT and Hadamard gates or the full Clifford group.  Four extra qubits are enough for fault-tolerant universal computation.} \label{f:computationqubits}
\end{figure}

In order to compute on data encoded within a single code block, we need to apply two- or three-qubit gates.  The particular circuits use symmetries of the codes or a more general round-robin construction from~\cite{YoderTakagiChuang16pieceableft}.  This is not fault tolerant, because a single gate failure can cause a correlated error of weight two or worse, which a distance-three code cannot correct.  To fix this, we replace each gate with a gadget involving two to four more qubits.  With no gate faults, the gadgets are equivalent to the ideal gates they replace.  The gadgets' purpose is to detect correlated errors, so that they can be corrected for later.  The gadgets cannot prevent the gates from spreading single-qubit faults into problematic multi-qubit errors.  To avoid this problem, we design the circuits carefully, and in some cases intersperse partial error correction procedures between gadgets, an idea from~\cite{KnillLaflammeZurek96} recently applied and extended by~\cite{HillFowlerWangHollenberg11codeconversion, YoderTakagiChuang16pieceableft}.  Sometimes error correction even needs to overlap the gadgets.  

For the basics of stabilizer algebra, quantum error-correcting codes and fault-tolerant quantum computation, we refer the reader to~\cite{Gottesman09faulttolerance}.

\section{Trick: Flagging a gate for correlated faults}

A main trick we use is to replace $\CZ$ and $\CCZ$ gates with small gadgets that can catch correlated faults.  The gadgets are reminiscent of one-ancilla-qubit fault-tolerant SWAP gate gadgets~\cite{Gottesman00local}.

\subsection{$\CZ$ gadget} \label{s:czflags}

The controlled-phase gate is a two-qubit diagonal gate $\CZ = \identity - 2 \ketbra{11}{11}$, represented in circuits as 
\begin{equation*}
\raisebox{-.2cm}{\includegraphics[scale=.769]{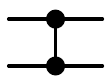}}
\end{equation*}
Observe that $\CZ$ gates commute with $Z$ errors, but copy $X$ (or $Y$) errors on one wire into $Z$ errors on the other: 
\begin{equation}
\raisebox{-.3cm}{
\raisebox{.135cm}{\includegraphics[scale=.769]{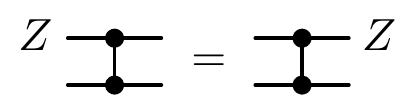}}
\quad
\includegraphics[scale=.769]{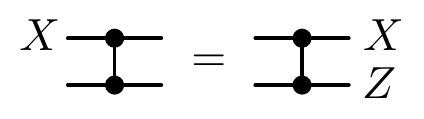}
}
\end{equation}

For fault tolerance, it suffices to study gates that fail with Pauli faults after the gate.  That is, when a noisy $\CZ$ gate fails, it applies the ideal $\CZ$ gate followed by one of the $15$ nontrivial two-qubit Pauli operators.  

\medskip

Consider the circuit of \figref{f:czxflag}, using one extra qubit that at the end is measured in the $\ket 0, \ket 1$ basis ($Z$ eigenbasis).  If the gates are perfect, then the measurement returns $0$ and this circuit has the effect of a single $\CZ$ gate.  If the $\CZ$ gate fails with an $X$ or $Y$ fault on the second qubit, however, then the measurement will return~$1$.  Thus certain kinds of faults can be detected.  If at most one location fails and the measurement returns~$0$, then the output cannot have an $XX$, $XY$, $YX$ or $YY$ error.  

\begin{figure}
\centering
\begin{tabular}{c@{$\quad\qquad$}c}
\subfigure[\label{f:czxflag}]{\raisebox{-.6cm}{\includegraphics[scale=.769]{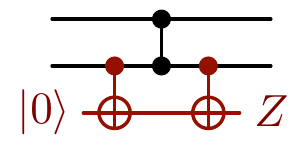}}} & 
\subfigure[\label{f:czzflag}]{\raisebox{-.6cm}{\includegraphics[scale=.769]{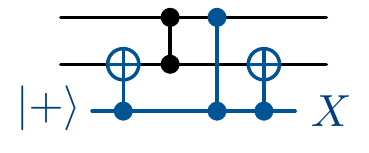}}} \\
\multicolumn{2}{c}{
\subfigure[\label{f:czflags}]{\raisebox{-.8cm}{\includegraphics[scale=.769]{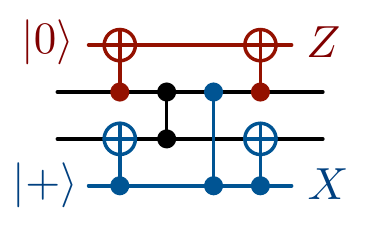}}}
}
\end{tabular}
\caption{$\CZ$ gadgets to catch correlated faults.  (a) An extra qubit can be used to catch $XX$, $XY$, $YX$ and $YY$ faults after the $\CZ$ gate.  (b) A similar circuit catches $ZZ$ faults.  (c) In combination, these gadgets can catch all two-qubit correlated faults (\thmref{t:czflags}).  
}
\end{figure}

A similar circuit can catch $Z$ faults.  If all gates in \figref{f:czzflag} are perfect, then the $X$ basis ($\ket +, \ket -$) measurement will return $+$ and the effect will be of a single $\CZ$.  If there is at most one fault and the measurement returns~$+$, then the output cannot have a $YX$ or $ZZ$ error.  (This fact can be verified by, for example, propagating $ZZI$ and $ZZX$ backward through the circuit, and observing that no single gate failure can create either.)  

The gadgets to catch $X$ and $Z$ faults can be combined: 

\begin{theorem} \label{t:czflags}
With no faults, the circuit of \figref{f:czflags} implements a $\CZ$ gate, with the measurements outputting $0$ and~$+$.  If there is at most one fault and the measurements return $0$ and~$+$, then neither $XX$, $XY$, $YX$, $YY$ nor $ZZ$ errors can occur on the output.  
\end{theorem}

Thus all single faults are caught except those equivalent to a fault on the $\CZ$ output or input qubits (namely, $IX, IY, IZ, XI, YI, ZI$ and $XZ, YZ, ZX, ZY$).  No $\CZ$ gadget can catch more errors than this.  

\begin{figure}
\centering
\begin{tabular}{c@{$\quad\qquad$}c}
\subfigure[\label{f:czzflagwrongorder}]{\raisebox{-.6cm}{\includegraphics[scale=.769]{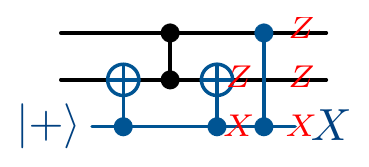}}} & 
\subfigure[\label{f:czflagswrongorder}]{\raisebox{-.6cm}{\includegraphics[scale=.769]{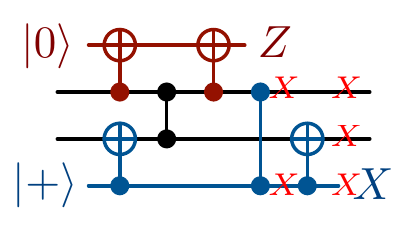}}}
\end{tabular}
\caption{Invalid $\CZ$ gadget constructions.  
} \label{f:czzflagwrongorderczflagswrongorder}
\end{figure}

The order of the gates matters.  Although the last two gates in \figref{f:czzflag} formally commute, switching them changes the faulty circuit so that an undetected $ZZ$ error can occur due to a single gate fault.  Similarly, in \figref{f:czflags} it is important that the gadget for catching $Z$ faults go inside that for $X$ faults.  With the other order, a single fault can lead to an undetected $XX$ error.  See \figref{f:czzflagwrongorderczflagswrongorder}.  

In practice, not every $\CZ$ gate need always be replaced with the full $\CZ$ gate gadget of \figref{f:czflags}.  Multiple gates can sometimes be combined under single flags.  We will see examples below.

\subsection{$\CCZ$ gadget}

The three-qubit gate $\CCZ = \identity - 2 \ketbra{111}{111}$ is denoted 
\begin{equation*}
\raisebox{-.4cm}{\includegraphics[scale=.769]{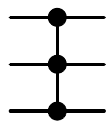}}
\end{equation*}
Again, it commutes with $Z$ errors.  It copies an $X$ error on the input into a $\CZ = \tfrac12 (II + IZ + ZI - ZZ)$ error on the output, i.e., into a linear combination of $II$, $IZ$, $ZI$ and $ZZ$ errors: 
\begin{equation*}
\raisebox{-.4cm}{\includegraphics[scale=.769]{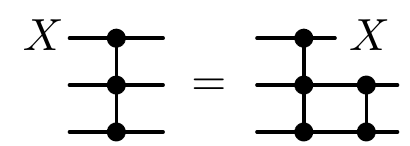}}
\end{equation*}

The following gadget, using four ancilla qubits, implements a $\CCZ$ on the black data qubits.  Furthermore, it satisfies that provided there is at most one failed gate and the measurement results are trivial, $0$ and $+$, then the error on the output is a linear combination of Paulis that could result from a one-qubit fault before or after a perfect $\CCZ$ gate, i.e., $III, ZII, XII, YII, XZI, YZI, XZZ, YZZ$ and qubit permutations thereof.  
\begin{equation} \label{e:cczgadget}
\raisebox{-1.4cm}{\includegraphics[scale=.769]{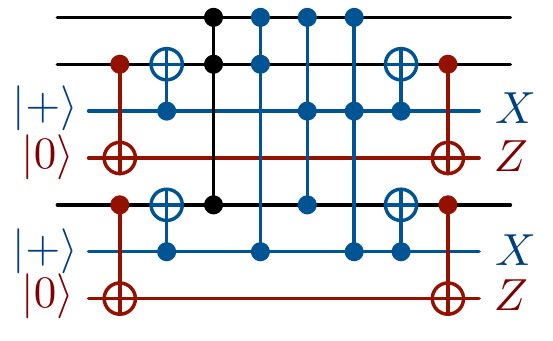}}
\end{equation}
That without faults the black and blue gates realize a $\CCZ$ is a special case of the following claim~\cite{YoderTakagiChuang16pieceableft}, with $r = 3$, $S^Z_1 = \{1\}$, $S^Z_2 = \{2, 3\}$, $S^Z_3 = \{5, 6\}$.  

\begin{claim} \label{t:logicalccz}
Consider an $n$-qubit CSS code with $k$ encoded qubits given by $\widebar X_j = X_{S^X_j}, \widebar Z_j = Z_{S^Z_j}$ for $S^X_j, S^Z_j \subseteq [n]$.  Let $U$ be the product of $\mathrm{C}^{(r-1)}Z = \identity - 2 \ketbra{1^r}{1^r}$ gates applied to every tuple of qubits in $S^Z_1 \times S^Z_2 \times \cdots \times S^Z_r$.  (If the $S^Z_j$ sets are not disjoint, then some of the applied gates will be $\mathrm{C}^{(s-1)}Z$ for some $s < r$.)  

Then $U$ is a valid logical operation.  It implements logical $\mathrm{C}^{(r-1)}Z$ on the first $r$ encoded qubits.  
\end{claim}

\begin{proof}
For the case that the sets $S^Z_j$ are disjoint, the claim is immediate in the case of singleton sets and follows in general by induction on the set size using the identity 
\begin{equation*} \label{e:cnotcommutecgate}
\raisebox{-.5cm}{\includegraphics[scale=.769]{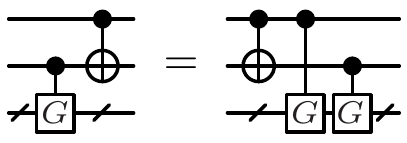}}
\end{equation*}
for any gate $G$, on one or more qubits, with $G^2 = \identity$.  (Figure~\ref{f:czzflag} provides one example.)  

More generally, the claim follows by writing out computational basis codewords as superpositions of computational basis states, and computing the effect of~$U$.  For $(x, y) \in \{0,1\}^r \times \{0,1\}^{k-r}$, the codeword $\ket{\widebar{x y}}$ is a uniform superposition of computational basis states: 
\begin{equation*}
\ket{\widebar{x y}} \propto \prod_{j=1}^r \widebar X_j^{x_j} \prod_{i=r+1}^k \widebar X_i^{y_{i-r}} \sum_{\text{stabilizers $X_S$}} X_S \ket{0^n}
\end{equation*}
The relationships $\widebar X_i \widebar Z_j = (-1)^{\delta_{ij}} \widebar Z_j \widebar X_i$ imply that $\abs{S^X_i \cap S^Z_j}$ is odd for $i = j$ and even otherwise; and for any stabilizer $X_S$, $[X_S, \widebar Z_j] = 0$ so $\abs{S \cap S^Z_j}$ is even.  

If $x \neq 1^r$, say $x_j = 0$, then for every term $\ket z$ in the above sum, $z$ has an even number of $1$s in $S^Z_j$, implying that $U \ket z = \ket z$.  Hence $U \ket{\widebar{x y}} = \ket{\widebar{x y}}$.  

If $x = 1^r$, then for any term $\ket z$ in $\ket{\widebar{x y}}$, $z$ has an odd number of $1$s in each $S^Z_j$ and therefore $U \ket z = - \ket z$.  Hence $U \ket{\widebar{x y}} = - \ket{\widebar{x y}}$.  
\end{proof}

\section{Fault-tolerant operations \mbox{for $\llbracket \lowercase{n},\lowercase{n}-2,2 \rrbracket$ codes}}

For even $n$, the $\llbracket n, n-2, 2 \rrbracket$ error-detecting code has stabilizers $X^{\otimes n}$ and $Z^{\otimes n}$, and logical operators $\widebar{X}_j = X_1 X_{j+1}$, $\widebar Z_j = Z_{j+1} Z_n$ for $j = 1, \ldots, n-2$.  This code, its symmetries, and methods of computing fault tolerantly on the encoded qubits were studied by Gottesman~\cite{Gottesman97}.  However, his techniques require at least $2 n$ extra qubits.  For example, to apply a CNOT gate between two logical qubits in the same code block, he teleports them each into separate code blocks, applies transversal CNOT gates between the blocks, and then teleports them back.  

We will give a fault-tolerant implementation of encoded CNOT and Hadamard gates on arbitrary logical qubits, using only two extra qubits.  Two-qubit fault-tolerant procedures for state preparation, error detection and projective measurement were given in~\cite{ChaoReichardt17errorcorrection}.  For $n \geq 6$ (so there are at least three encoded qubits), we will give a four-qubit fault-tolerant implementation of the encoded $\CCZ$ gate, thereby completing a universal gate set.

\subsection{Permutation symmetries and transversal operations}

Fault tolerance for a distance-two code means that any single fault within an operation should either have no effect or lead to a detectable error.  For example, of course the $4^{n-2}$ logical Pauli operators can all be applied fault tolerantly, since the operations do not couple any qubits.  

All qubit permutations preserve the two stabilizers and therefore preserve the code space.  They are also fault tolerant if implemented either by relabeling the qubits or by moving them past each other (and not by using two-qubit SWAP gates~\cite{Gottesman00local}).  For $i, j \in [n-2]$, $i \neq j$, the qubit swap $(i+1,j+1)$ swaps the logical qubits~$i$ and~$j$.  The qubit swap $(1, 2)$ implements logical CNOTs from qubits $2$ through $n - 2$ into~$1$, and the qubit swap $(2, n)$ implements logical CNOTs in the opposite direction: 
\begin{equation*}
\raisebox{-.8cm}{\includegraphics[scale=.769]{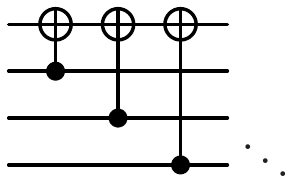}}
\raisebox{-.1cm}{\quad\text{and}\qquad}
\raisebox{-.8cm}{\includegraphics[scale=.769]{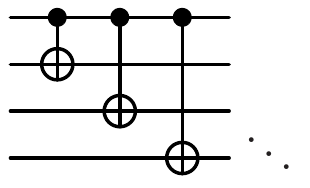}}
\end{equation*}

Transversal Hadamard, $H^{\otimes n}$, followed by the qubit swap $(1, n)$, implements logical $H^{\otimes (n-2)}$.  

The Clifford reflection $G = \tfrac{i}{\sqrt 2}(X + Y)$ conjugates $X \leftrightarrow Y$ and $Z \rightarrow -Z$.  Transversal $G$ is a valid logical operation (up to $Z_n$ to correct the $X^{\otimes n}$ syndrome if $n = 2 \mod 4$).  It implements logical $\CZ$ gates between all encoded qubits: 
\begin{equation*}
\raisebox{-.8cm}{\includegraphics[scale=.769]{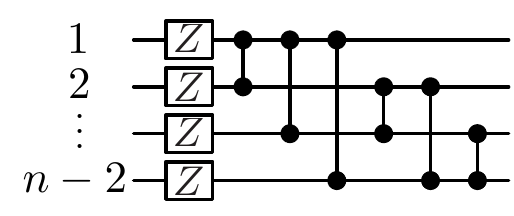}}
\end{equation*}

The operations given so far generate a group much smaller than the full $(n-2)$-qubit Clifford group.  The qubit permutations generate $n!$ different logical operations (except for $n = 4$, just $6$ operations).  With the transversal application of the six one-qubit Clifford gates, up to Paulis, this gives $6 (n!)$ different logical operations (or $36$ for $n = 4$).  \tabref{f:groupsizes} gives the sizes of various interesting subgroups of the Clifford group, for comparison.  

\begin{table}
\begin{center}
\begin{tabular}{c |@{\quad} c @{\quad} l @{\quad} l@{\quad} l}
\hline \hline
$k$ & $S_k$ & $\GL(k, 2)$ & $\langle \text{CNOT}, H \rangle$ & $\mathcal C_k / \mathcal P_k$ \\
\hline
1 & 1 & $\times 1$ & $\times 2$ & $\times 3$ \\
2 & 2 & $\times 3$ & $\times 12$ & $\times 10$ \\
3 & 6 & $\times 28$ & $\times 240$ & $\times 36$ \\
4 & 24 & $\times 840$ & $\times 17280$ & $\times 136$ \\
5 & 120 & $\times 83328$ & $\times 4700160$ & $\times 528$ \\
6 & 720 & $\times 27998208$ & $\times 4963368960$ & $\times 2080$ \\
7 & 5040 & $\times 3.2 \cdot 10^{10}$ & $\times 2.1 \cdot 10^{13}$ & $\times 8256$ \\ 
8 & 40320 & $\times 1.3 \cdot 10^{14}$ & $\times 3.4 \cdot 10^{17}$ & $\times 32896$ \\
\hline \hline
\end{tabular}
\end{center}
\caption{Sizes of the $k$-qubit Clifford group and subgroups.  There are $\abs{S_k} = k!$ permutations of $k$ qubits.  CNOT gates generate a group of size $\abs{\GL(k,2)} = \prod_{j=0}^{k-1}(2^k - 2^j)$, adding Hadamard gates generates a larger group, and finally the full Clifford group, up to the $\abs{\mathcal P_k} = 4^k$ Paulis, has size $2^{k^2} \prod_{j=1}^k (4^k - 1)$.  The sizes are given as multiples of the previous columns, e.g., $\abs{\mathcal C_2 / \mathcal P_2} = 2 \times 3 \times 12 \times 10 = 720$.  \comment{$\abs{\langle \text{CNOT}, H \rangle} = \abs{\mathcal C_k / \mathcal P_k} / (2^{k-1} (2^k + 1))$.}  
} \label{f:groupsizes}
\end{table}

We next give fault-tolerant implementations for a logical Hadamard gate on a single encoded qubit, and for a logical $\CZ$ gate between two encoded qubits.  These generate a large subgroup of the Clifford group, the $\langle \text{CNOT}, H \rangle$ column in \tabref{f:groupsizes}.

\subsection{$\CZ$ gate}

By \claimref{t:logicalccz}, a logical $\CZ_{1,2}$ gate can be implemented by $Z_n \, \CZ_{2,3} \, \CZ_{2,n} \, \CZ_{3,n}$: 
\begin{equation} \label{e:cz23ncz12}
\text{physical}
\raisebox{-.5cm}{\includegraphics[scale=.769]{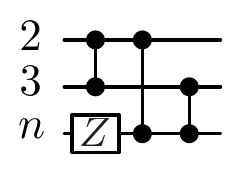}}
\;\, = \;\,
\text{logical}
\raisebox{-.3cm}{\includegraphics[scale=.769]{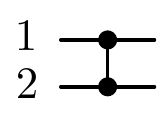}}
\end{equation}

However, this implementation is not fault tolerant.  Some failures are detectable; for example, if the $\CZ_{2,3}$ gate fails as $XI$, then the final, detectable error is $X_2 Z_n$.  Others are not, e.g., if the $\CZ_{2,3}$ gate fails as $XX$, then the final $X_2 X_3 = \widebar X_1 \widebar X_2$ error is undetectable.  

The bad faults that can cause undetectable logical errors are as follows: 
\begin{center}
\begin{tabular}{c c | c c | c c}
\hline \hline
\multicolumn{2}{c|}{$\CZ_{2,3}$} & \multicolumn{2}{c|}{$\CZ_{2,n}$} & \multicolumn{2}{c}{$\CZ_{3,n}$} \\
Fault & Error & Fault & Error & Fault & Error \\
\hline
$ZZ$ & $Z_2 Z_3$
& $ZZ$ & $Z_2 Z_n$
& $ZZ$ & $Z_3 Z_n$ \\
$XX$ & $X_2 X_3$
& $XY$ & $X_2 Z_3 Y_n$
& $XX$ & $X_3 X_n$ \\
$YY$ & $Y_2 Y_3$
& $YX$ & $Y_2 Z_3 X_n$
& $YY$ & $Y_3 Y_n$ \\
\hline \hline
\end{tabular}
\end{center}
In particular, all these bad faults are caught by the $\CZ$ gadget of \thmref{t:czflags}.  Therefore replacing each physical $\CZ$ gate in~\eqnref{e:cz23ncz12} with that gadget gives a fault-tolerant implementation of a logical $\CZ_{1,2}$ gate.  The circuit uses at most two ancilla qubits at a time.  

In fact, one can simplify the resulting circuit by using the same $\ket 0$ ancilla to catch $X$ faults on multiple $\CZ$ gates.  The following circuit is also fault tolerant: 
\begin{equation*}
\raisebox{-1.3cm}{\includegraphics[scale=.769]{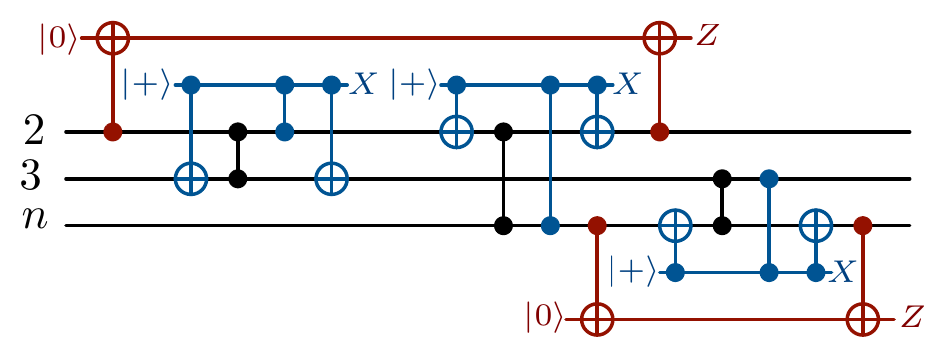}}
\end{equation*}

The gadgets to catch $Z$ faults can be merged, too.  The following circuit is fault tolerant, and still requires at most two ancilla qubits at a time: 
\begin{equation}
\raisebox{-1.3cm}{\includegraphics[scale=.769]{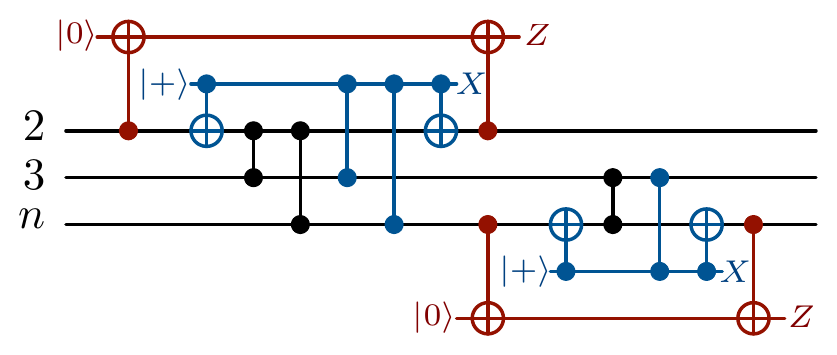}}
\end{equation}
Perhaps further simplifications are possible.

\vspace{-.1cm}
\subsection{Targeted Hadamard gate}

A single encoded Hadamard gate can also be implemented fault tolerantly with two extra qubits.  The black portion of the circuit below, with \raisebox{-.1cm}{\includegraphics[scale=.45]{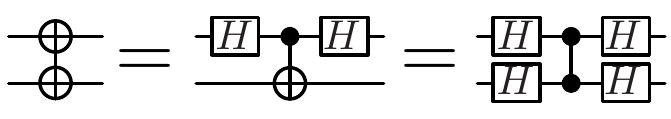}}$\;$, implements $\widebar H_1$.  The red and blue portions, analogous to Figs.~\ref{f:czxflag} and~\ref{f:czzflag}, respectively, catch problematic faults.  $X$ measurements should return~$+$ and $Z$ measurements~$0$.  
\begin{equation}
\raisebox{-1.3cm}{\includegraphics[scale=.769]{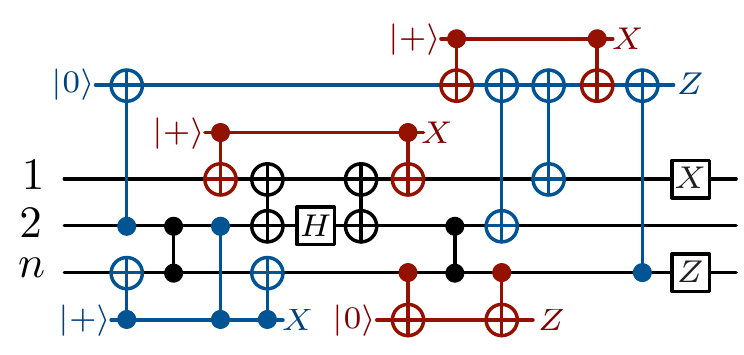}}
\end{equation}
The circuit's fault tolerance can be verified by enumerating all the ways in which single gates can fail.

\vspace{-.1cm}
\subsection{Four-ancilla $\CCZ$ gate}

For $n \geq 6$, a $\CCZ$ gate on encoded qubits $1, 2, 3$ can be implemented by round-robin $\CCZ$ gates on $\{2, n\} \times \{3, n\} \times \{4, n\}$, by \claimref{t:logicalccz}: 
\begin{equation} \label{e:errordetectingroundrobinccz234n}
\raisebox{-.8cm}{\includegraphics[scale=.769]{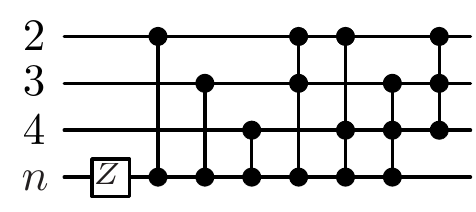}}
\end{equation}
This circuit uses one $Z$, three $\CZ$ and four $\CCZ$ gates.  To make it fault tolerant, use the gadget from \figref{f:czflags} for each $\CZ$ gate, and replace each $\CCZ$ with the gadget of Eq.~\eqnref{e:cczgadget}.  Overall, this requires four ancilla qubits.  

Single gate faults are either caught by the gadgets or lead to an error that could also arise from a one-qubit fault between the gates in~\eqnref{e:errordetectingroundrobinccz234n}.  A one-qubit $X$ or $Y$ fault will be detectable at the end because it is copied only to linear combinations of $Z$s---the $X$ component of the final error will still have weight one---and a one-qubit $Z$ fault will be detectable because it commutes through the $\CZ$ and $\CCZ$ gates.  Therefore the procedure is fault tolerant.

\vspace{0cm}
\section{Fault-tolerant operations for the $\llbracket 15,7,3 \rrbracket$ \mbox{Hamming code}}

The $\llbracket 15,7,3 \rrbracket$ Hamming code is a self-dual CSS, perfect distance-three code.  Packing seven logical qubits into $15$ physical qubits, it is considerably more efficient than more commonly used $\llbracket 7,1,3 \rrbracket$ and $\llbracket 9,1,3 \rrbracket$ CSS codes, although it tolerates less noise.  

We first give a presentation of the code and its symmetries following~\cite{Harrington11permutations}.  Then we give a two-ancilla-qubit method for fault tolerantly implementing the full Clifford group on the encoded qubits, and, to complete a universal gate set, a four-qubit fault-tolerant encoded $\CCZ$ gate.  

Two-qubit fault-tolerant procedures for state preparation and error correction were given in~\cite{ChaoReichardt17errorcorrection}.

\vspace{0cm}
\subsection{$\llbracket 15,7,3 \rrbracket$ Hamming code}

The $\llbracket 15,7,3 \rrbracket$ Hamming code has four $X$ and four $Z$ stabilizers each given by the following parity-checks: 
\begin{equation} \label{e:1573stabilizers}
\begin{tabular}{c c c c c c c c c c c c c c c}
0&0&0&0&0&0&0&1&1&1&1&1&1&1&1\\
0&0&0&1&1&1&1&0&0&0&0&1&1&1&1\\
0&1&1&0&0&1&1&0&0&1&1&0&0&1&1\\
1&0&1&0&1&0&1&0&1&0&1&0&1&0&1
\end{tabular}
\end{equation}
Index the qubits left to right from $1$ to $15$.  Observe that the columns are these numbers in binary.  

As in~\cite{Harrington11permutations}, we define logical operators based on the following seven weight-five strings: 
\begin{equation} \label{e:1573logicaloperators}
\begin{tabular}{c c c c c c c c c c c c c c c}
1&1&0&1&0&0&0&1&0&0&0&0&0&0&1\\
1&1&0&0&1&0&0&0&0&1&0&1&0&0&0\\
1&1&0&0&0&1&0&0&0&0&1&0&0&1&0\\
1&1&0&0&0&0&1&0&1&0&0&0&1&0&0\\
1&0&0&1&0&1&0&0&1&1&0&0&0&0&0\\
1&0&0&1&0&0&1&0&0&0&0&1&0&1&0\\
1&0&0&0&0&0&0&1&0&1&0&0&1&1&0
\end{tabular}
\end{equation}
From the first string, $\widebar X_1 = XXIXIIIXIIIIIIX$ and $\widebar Z_1 = ZZIZIIIZIIIIIIZ$.  The remaining strings specify the logical operators $\widebar X_2, \widebar Z_2$ through $\widebar X_7, \widebar Z_7$.

\subsection{Transversal operations}

Transversal operations are automatically fault tolerant.  

Transversal Pauli operators implement logical transversal Pauli operators.  Indeed, transversal $X$, i.e., $X^{\otimes 15}$, preserves the code and implements transversal logical $X$, i.e., $X^{\otimes 7}$, on the code space, and similarly for $Y$ and~$Z$.  

In fact, any one-qubit Clifford operator applied transversally preserves the code space and implements the same operator transversally on the encoded qubits.  For example, since the logical operators are each self-dual, applying the Hadamard gates $H^{\otimes 15}$ implements logical $H^{\otimes 7}$.  

Of course, since the code is CSS, transversal CNOT gates between two code blocks implements transversal logical CNOT gates on the code spaces.  Furthermore, \cite{PaetznickReichardt13universal} shows that on three code blocks transversal CCZ can be used to obtain a universal gate set.  Here, however, we will consider only single code blocks and Clifford operations.

\subsection{Permutation symmetries}

Permutations of the qubits are also fault tolerant, either by physically moving the qubits or by relabeling them.  

The code's permutation automorphism group has order $20,\!160$, and is isomorphic to $A_8$ and $\GL(4,2)$~\cite{Harrington11permutations, Grassl13automorphisms}.  It is generated by the following three $15$-qubit permutations: 
\begin{equation} \label{e:1573permutationautomorphisms}
\begin{split}
\sigma_1 &= (1,2,3)(4,14,10)(5,12,9)(6,13,11)(7,15,8) \\
\sigma_2 &= (1,10,5,2,12)(3,6,4,8,9)(7,14,13,11,15) \\
\sigma_3 &= (1,10,15,3,8,13)(4,6)(5,12,11)(7,14,9)
\end{split}\end{equation}
These permutations fix the code space, but act nontrivially within it.  The permutations $\sigma_1$ and $\sigma_2$ apply the respective permutations $(1,2,3)$ and $(3,4,5,6,7)$ to the seven logical qubits.  Together, these generate the alternating group $A_7$ of even permutations.  

The logical effect of $\sigma_3$ is not a permutation.  It is equivalent to the following circuit of $24$ CNOT gates, in which gates with the same control wire are condensed: 
\begin{equation}
\raisebox{-1.2cm}{\includegraphics[scale=.577]{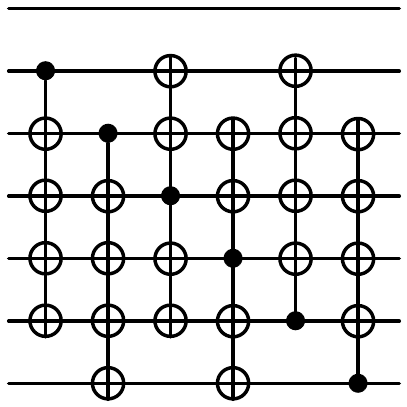}}
\end{equation}
Thus the first logical qubit is fixed, while for $j \in \{2, \ldots, 7\}$ and $P \in \{X, Y, Z\}$, $P_j$ is mapped to $(\prod_{j=2}^7 P_j) P_{j+1}$, wrapping the indices cyclically.  This is a six-qubit generalization of a four-qubit operator studied in~\cite[Sec.~6]{Gottesman97}.  (Like permutations, this operation has the property of being a valid transversal operation on any stabilizer code.)

\subsection{$\CZ$ circuits based on permutation symmetries} \label{s:cz89to1011and1213to1415}

Any permutation symmetry of the code can be turned into a $\CZ$ automorphism (\figref{f:permutationcz}): 

\begin{figure}
\centering
\begin{tabular}{c@{$\qquad$}c}
\subfigure[\label{f:permutationcz}]{\raisebox{-1.5cm}{\includegraphics[scale=.769]{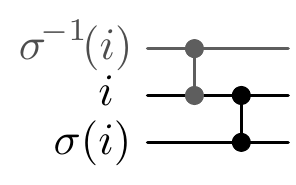}}} & 
\subfigure[\label{f:czsigma3logical}]{\raisebox{-1.5cm}{\includegraphics[scale=.769]{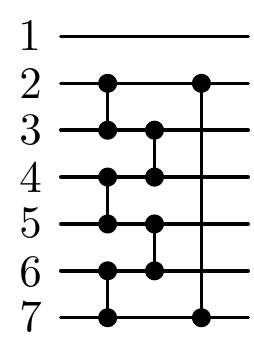}}}
\end{tabular}
\caption{(a) If $\sigma$ is a permutation automorphism for a self-dual CSS code, then applying $\CZ$ gates along its cycles fixes the code space (\claimref{t:permutationcz}).  (b) The logical effect for $\CZ$ gates following the permutation $\sigma_3$.}
\end{figure}

\begin{claim} \label{t:permutationcz}
For a self-dual CSS code, if $\sigma$ is a qubit permutation that fixes the code space, then the circuit with a $\CZ$ gate from $i$ to $\sigma(i)$, for all $i \neq \sigma(i)$, fixes the code space up to Pauli $Z$ corrections.  
\end{claim}

\begin{proof}
$Z$ stabilizers commute with the $\CZ$ gates, so are preserved.  An $X$ stabilizer $X_S = \prod_{i \in S} X_i$ is mapped to $\prod_{i \in S} (X_i Z_{\sigma(i)} Z_{\sigma^{-1}(i)}) = \pm X_S Z_{\sigma(S) \cup \sigma^{-1}(S)}$.  Up to sign, this is a stabilizer, since $Z_{\sigma(S)} Z_{\sigma^{-1}(S)}$ is a stabilizer.  
\end{proof}

For example, the physical circuit in Eq.~\eqnref{e:cz23ncz12} comes from the cyclic permutation $(2,3,n)$ of the $\llbracket n, n-2, 2 \rrbracket$ code.  

Applying \claimref{t:permutationcz} to $\sigma_3$ of Eq.~\eqnref{e:1573permutationautomorphisms}, the two $\CZ$ gates for the cycle $(4, 6)$ cancel out, leaving the gates $(\CZ_{1,10} \CZ_{10,15} \cdots \CZ_{13,1}) (\CZ_{5,12} \CZ_{12,11} \CZ_{11,5}) \ldots$.  As shown in \figref{f:czsigma3logical}, 
the effect is that of logical $\CZ$ gates following the cycle $(2,3,4,5,6,7)$.  

Notice that the logical effect necessarily consists of encoded $\CZ$ gates, because logical $Z$ operators are unchanged and logical $X$ operators pick up $Z$ components.  Also, the map from \claimref{t:permutationcz} is not a homomorphism from permutations into unitary circuits.

\subsubsection*{$\CZ$ gates $\{8,9\}$ to $\{10,11\}$ and $\{12,13\}$ to $\{14,15\}$}

The permutation $(6,7)(8,10,9,11)(12,14,13,15)$ fixes the code, and under \claimref{t:permutationcz} corresponds to the eight $\CZ$ gates of \figref{f:cz89to1011and1213to1415}.  Figure~\ref{f:cz89to1011and1213to1415logical} gives their logical effect.  

\begin{figure}
\centering
\begin{tabular}{c@{$\quad$}c}
\subfigure[\label{f:cz89to1011and1213to1415}]{\raisebox{-1.5cm}{\includegraphics[scale=.769]{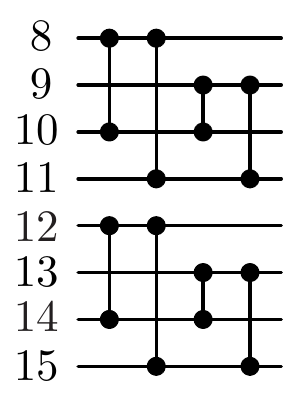}}} & 
\subfigure[\label{f:cz89to1011and1213to1415logical}]{\raisebox{-1.5cm}{\includegraphics[scale=.769]{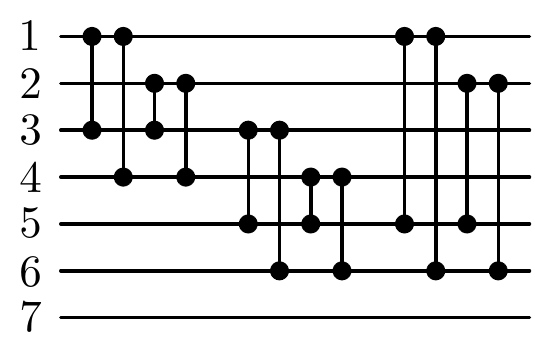}}} \\
\multicolumn{2}{c}{\subfigure[\label{f:cz89to1011and1213to1415errorcorrection}]{
\raisebox{-3.5cm}{\hspace{-.5cm}\includegraphics[scale=.62]{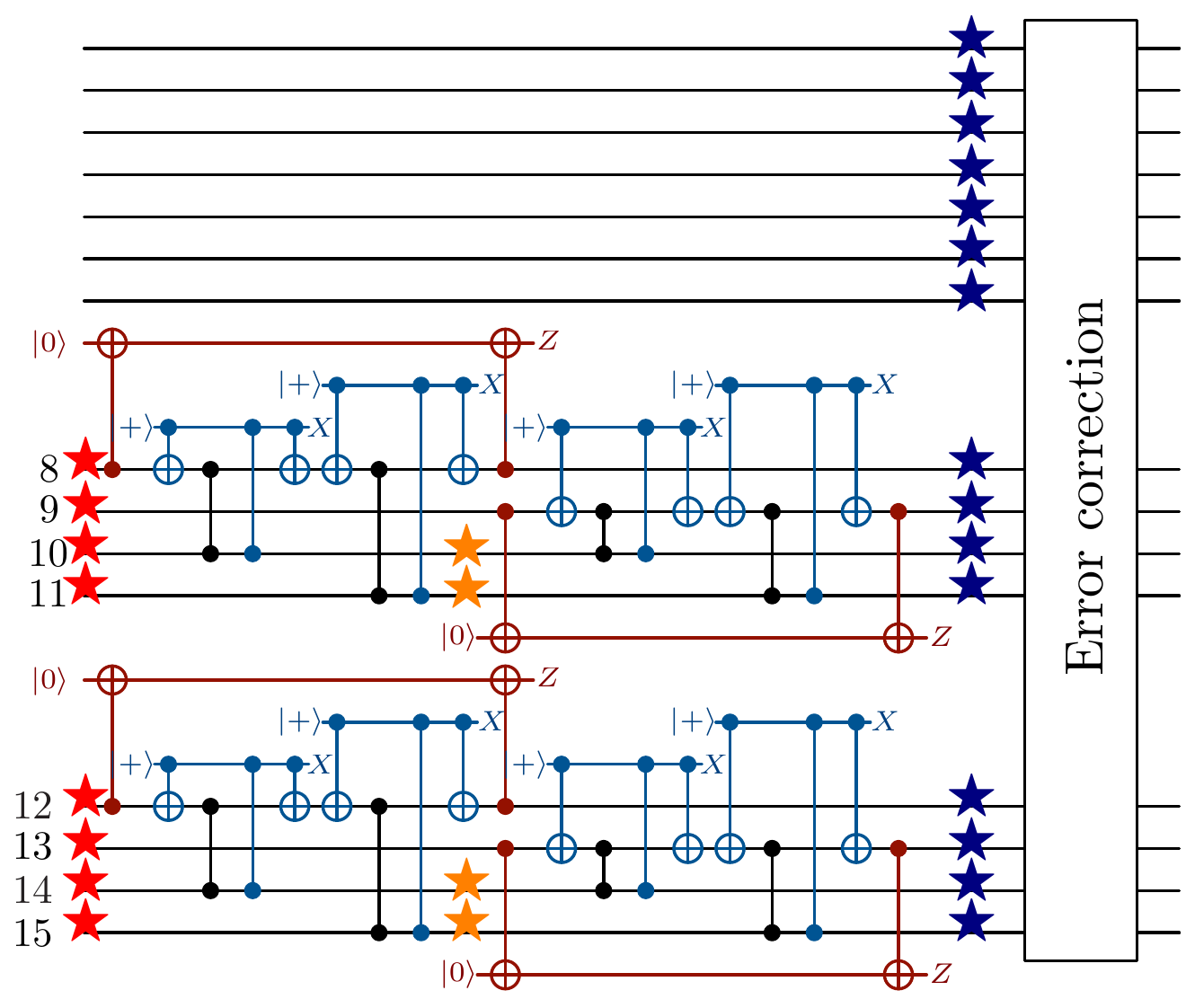}}
}}
\end{tabular}
\caption{The physical $\CZ$ gates in (a) have the logical effect of (b).  (c) The circuit can be made fault tolerant by replacing the $\CZ$ gates with the gadgets from \secref{s:czflags}, and simplifying.} \label{f:cz89to1011and1213to1415physicallogical}
\end{figure}

Using Magma~\cite{Magma97}, we compute that the group generated by this operation, the permutations $\sigma_1, \sigma_2, \sigma_3$, and transversal $H$ has the same size as the $\langle \text{CNOT}, H \rangle$ group on \emph{six} qubits (about $1.001 \cdot 10^{20}$).  (Adding transversal $G$ only triples the group size.)  This hints that by working in a logical basis in which one qubit has $\widebar X = X^{\otimes 15}$ and $\widebar Z = Z^{\otimes 15}$ (operators fixed by the permutations and by \figref{f:cz89to1011and1213to1415}), perhaps arbitrary combinations of CNOT and $H$ can be applied to the other six qubits.  But no, only half the $\langle \text{CNOT}, H \rangle$ group can be reached.  

The circuit of \figref{f:cz89to1011and1213to1415} is clearly not fault tolerant.  (For example, an $XX$ fault after the $\CZ_{9,11}$ gate gives the error $X_9 X_{11}$, which is indistinguishable from $X_2$.)  We can use the gadgets from \secref{s:czflags} for each $\CZ$ gate to obtain the circuit of \figref{f:cz89to1011and1213to1415errorcorrection}, shown with the trailing error correction.  Two ancilla qubits are needed.  

We claim that this compiled circuit is fault tolerant.  This means that if the input lies in the code space, the compiled circuit has at most one fault (a two-qubit Pauli fault after a gate, or a one-qubit fault on a resting qubit), and the subsequent error correction is perfect; then the final outputs lie in the code space with no logical errors.  To verify fault tolerance, there are two cases to check.  

\medskip

First, consider the case that, with at most one fault, all the gadget measurements give the trivial output ($0$ for a $Z$ measurement, $+$ for $X$).  Since the gadgets catch two-qubit gate faults, we need only check possible one-qubit faults between gates.  Inequivalent fault locations are marked with stars in \figref{f:cz89to1011and1213to1415errorcorrection}.  (Faults at other locations either cause the same errors, or will be caught.)  In particular, entering error correction the possible error can be $\identity$, $X_1, Z_1, Y_1, \ldots, X_{15}, Z_{15}, Y_{15}$---from the ${\color{Midnight} \bigstar}$ locations.  Or, from ${\color{red} \bigstar}$ locations, it can be $XIZZ, YIZZ, IXZZ, IYZZ, ZZXI, ZZYI, ZZIX$, $ZZIY$, and $IZXI, IZYI, IZIX, IZIY$ from ${\color{Tangerine} \bigstar}$ locations, on qubits $8, 9, 10, 11$---and similarly for qubits $12$ to $15$.  This give $70$ different errors total.  All $70$ have distinct syndromes, and therefore can be corrected.  (This fact can be verified either by computing all the syndromes, or by observing from Eq.~\eqnref{e:1573stabilizers} that $Z_8 Z_9 Z_{10} Z_{11}$, $Z_{12} Z_{13} Z_{14} Z_{15}$, $Z_1 Z_6 Z_7$, $Z_1 Z_8 Z_9$, $Z_1 Z_{12} Z_{13}$ and $Z_1 Z_{14} Z_{15}$ are logical operators.  Thus, for example, if you observe the $Z$ syndrome $1000$, for error $X_8$, and the $X$ syndrome $0001$, for error~$Z_1$, you can safely correct $X_8 Z_{10} Z_{11}$.  $Z_1 X_8$ cannot occur.)  

Note that the error-correction procedure needs to take into account the $X$ and $Z$ stabilizer syndromes together to decide what correction to apply.  This can work because the $\llbracket 15,7,3 \rrbracket$ code is not a perfect stabilizer code: there are $2^8$ possible syndromes but only $1 + 15 \cdot 3$ possible trivial or weight-one errors.  (It is only perfect as a CSS code, i.e., the $2^4$ $Z$ stabilizer syndromes are exactly enough to correct the $1 + 15$ possible trivial or weight-one $X$ errors, and similarly for $Z$ errors.)  This leaves room to correct some errors of weight more than one.  

\medskip

Next, consider the case that, with at most one fault in \figref{f:cz89to1011and1213to1415errorcorrection}, one or more of the gadget measurements gives a nontrivial output.  This case is much simpler, because the measurement results localize the fault, leaving only a few possibilities for the error entering error correction.  One must verify that in all cases, these possibilities are distinguished by their syndromes.  

For example, if the first $Z$ measurement returns $1$ and all other measurements are trivial, the errors from single faults that can occur are, on qubits $8, 9, 10, 11$: 
\begin{equation*}
\begin{array}{r@{,\;} r@{,\;} r@{,\;} r}
IIII & ZIII & XIII & YIII, \\
XIIZ &YIIZ & XIZZ & YIZZ, \\
XZIX & XZIY & XZXZ & XZYZ, \\
YZIX & YZIY & YZXZ & YZYZ
\end{array}
\end{equation*}
These $16$ possible errors all have distinct syndromes, so are correctable.  

As another example, if the last two measurements, of qubits coupled to qubit~$13$, are nontrivial, then the possible errors from single faults are, on qubits $12, 13, 14, 15$: 
\begin{equation*}
IXII, IYII, IXIZ, IYIZ, IXIX, IXIY
\end{equation*}
Again, these have distinct syndromes.  

Other possible measurement outcomes are similar.  We have used a computer to check them all.

\subsection{Round-robin $\CZ$ circuits to complete the Clifford group}

The above operations do not generate the full seven-qubit logical Clifford group, and we have not been able to find a permutation for which applying \claimref{t:permutationcz} enlarges any further the generated logical group.  Instead, we turn to the round-robin construction of \claimref{t:logicalccz}.

\subsubsection*{$\CZ$ gates $4$ to $\{5,6,7\}$, $8$ to $\{9,10,11\}$, $12$ to $\{13,14,15\}$}

Observe that $Z_{\{4,8,12\}}$ and $Z_{\{4,5,6,7\}} \sim Z_{\{8,9,10,11\}} \sim Z_{\{12,13,14,15\}}$ are logical operators, implementing respectively $\widebar Z_{\{2,5,7\}}$ and $\widebar Z_{\{1,2,3,4\}}$.  By a minor extension of \claimref{t:logicalccz}, applying $Z_{\{4,8,12\}}$ and nine $\CZ$ gates from $4$ to each of qubits $\{5,6,7\}$, $8$ to $\{9,10,11\}$, and $12$ to $\{13,14,15\}$ preserves the code space.  The logical effect is 
\begin{equation}
\raisebox{-1.3cm}{\includegraphics[scale=.769]{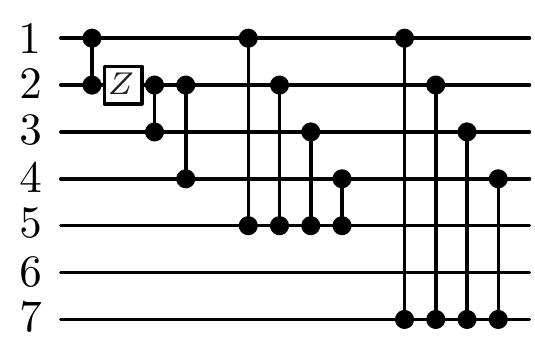}}
\end{equation}

Together with permutations and transversal operations, this circuit completes the seven-qubit Clifford group, without needing the operation from \secref{s:cz89to1011and1213to1415}.  To make the operation fault tolerant, we will transform it in three~steps.  

First, consider the circuit below, in which we have wrapped the $\CZ$ gates leaving qubits $4$, $8$ and $12$ with overlapping gadgets to catch $X$ faults.  If at most one fault occurs and one or more of the $Z$ measurements gives~$1$, then the errors that can occur are distinguished by their syndromes.  
\begin{equation} \label{e:cz4to567and8to91011and12to131415flags}
\raisebox{-3.7cm}{\includegraphics[scale=.769]{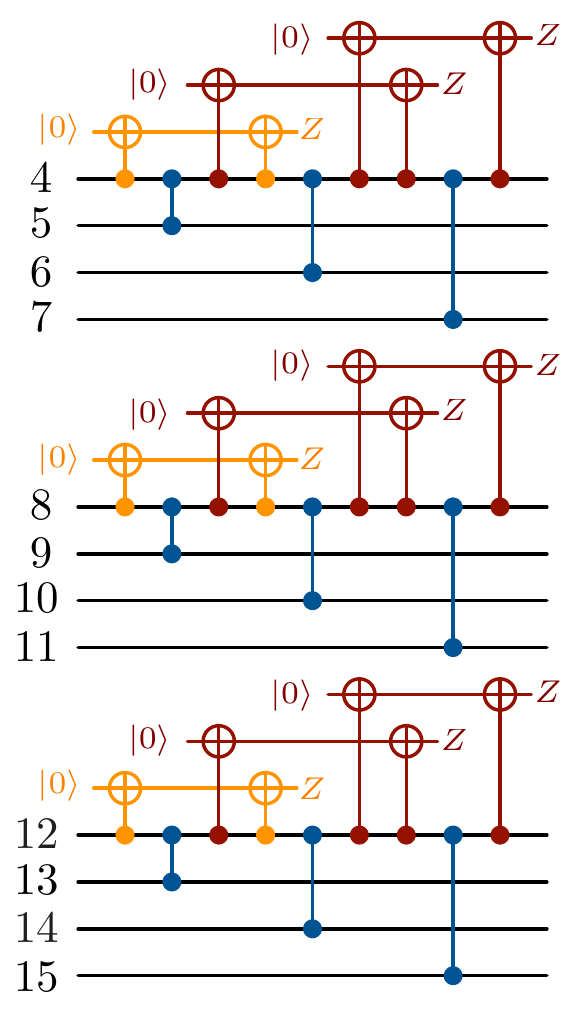}}
\end{equation}
For example, if the orange, first $Z$ measurement off qubit~$4$ gives~$1$ and all others~$0$, then the possible errors entering error correction are, on qubits $4, 5, 6, 7$: 
\begin{equation*}
\begin{array}{r@{,\;} r@{,\;} r@{,\;} r}
IIII & ZIII \\
XIZZ & XXZZ & XYZZ & XZZZ \\
YIZZ & YXZZ & YYZZ &YZZZ 
\end{array}
\end{equation*}
These errors all have distinct syndromes.  

The circuit in Eq.~\eqnref{e:cz4to567and8to91011and12to131415flags} is not fault tolerant, however, because with at most one fault if all the $Z$ measurements give~$0$, some inequivalent errors will have the same syndrome.  We can list the problematic errors.  For each of the following sets, the errors within the set are all possible, but have the same syndrome: 
\begin{equation}\begin{gathered} \label{e:baderrors}
\{ Z_1, Z_4 Z_5, Z_8 Z_9, Z_{12} Z_{13} \} \\
\{ Z_2, Z_4 Z_6, Z_8 Z_{10}, Z_{13} Z_{14} \} \\
\{ Z_3, Z_4 Z_7, Z_8 Z_{11}, Z_{12} Z_{15} \} \\
\{ X_4, Y_4 Z_5 Z_6 Z_7 \}, \{ Y_4, X_4 Z_5 Z_6 Z_7 \} \\
\{ X_8, Y_8 Z_9 Z_{10} Z_{11} \}, \{Y_8, X_8 Z_9 Z_{10} Z_{11} \} \\
\{ X_{12}, Y_{12} Z_{13} Z_{14} Z_{15} \}, \{ Y_{12}, X_{12} Z_{13} Z_{14} Z_{15} \}
\end{gathered}\end{equation}

Next replace each of the blue $\CZ$ gates in Eq.~\eqnref{e:cz4to567and8to91011and12to131415flags} with the $ZZ$ fault gadget from \figref{f:czzflag}.  This gadget has the property that, with at most one failure, a $ZZ$ fault can only occur if the $X$ measurement returns~$-$.  These measurements thus distinguish the errors in the first three sets above.  

Yet the new circuit is still not fault tolerant.  The gadget measurements cannot distinguish the errors in each of the last six sets in~\eqnref{e:baderrors}.  For example, consider an $X_4$ error before the circuit.  It propagates to $X_4 Z_5 Z_6 Z_7$.  Since $Z_4 Z_5 Z_6 Z_7$ is a logical error, $X_4 Z_5 Z_6 Z_7$ is indistinguishable from a $Y_4$ error after the circuit, and no error-correction rules can correct for the possible logical error.  

Gadgets cannot protect against single-qubit faults that occur just before or after the circuit.  This circuit is qualitatively different from the one we considered in \secref{s:cz89to1011and1213to1415}, and a new trick is needed to make it fault tolerant.  

\medskip

Consider the following circuit from~\cite{ChaoReichardt17errorcorrection}, ignoring for now the orange portion at top right.  
\begin{equation} \label{e:456712131415ftsyndrome}
\raisebox{-2.4cm}{\includegraphics[scale=.769]{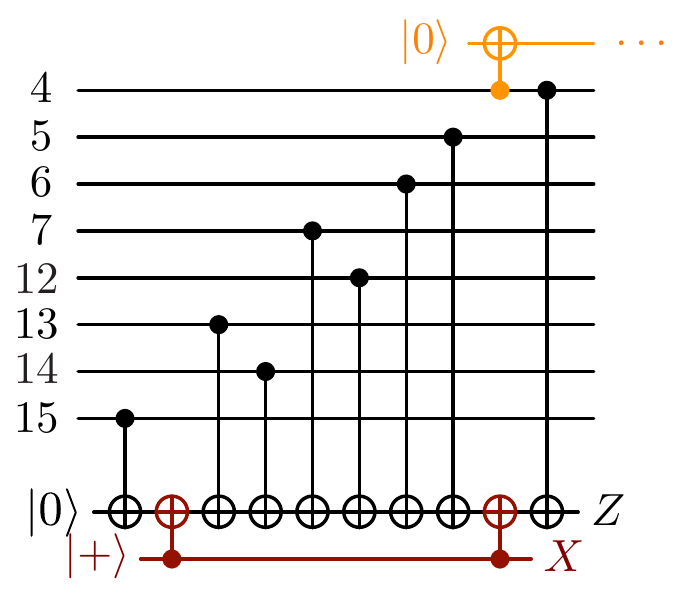}}
\end{equation}
This circuit fault-tolerantly extracts the syndrome of the $Z_{\{4, 5, 6, 7, 12, 13, 14, 15\}}$ stabilizer, in the sense that: 
\begin{itemize}
\item With no gate faults, the $Z$ measurement gives the syndrome, and the $X$ measurement gives~$+$.  
\item With at most one fault, if the $X$ measurement gives~$+$, then the data error has weight $\leq 1$.  
\item With at most one fault, if the $X$ measurement gives~$-$, then the $X$ component of the data error has weight $\leq 1$.  The $Z$ component can be any of 
\begin{equation*}\begin{gathered}
Z_4, Z_{\{4, 5\}}, Z_{\{4, 5, 6\}}, Z_{\{4, 5, 6, 12\}}, Z_{\{4, 5, 6, 7, 12\}} \sim Z_{\{13, 14, 15\}}, \\
Z_{\{4, 5, 6, 7, 12, 14\}} \sim Z_{\{13, 15\}}, Z_{\{4, 5, 6, 7, 12, 13, 14\}} \sim Z_{15}, \identity
\end{gathered}\end{equation*}
and these errors all have distinct syndromes.  (The order of the CNOT gates ensures this property.)  
\end{itemize}
In~\cite{ChaoReichardt17errorcorrection}, this circuit was used in a two-ancilla-qubit fault-tolerant error-correction procedure.  

The circuit in~\eqnref{e:456712131415ftsyndrome} is useful for us now to detect an $X_4$ or $Y_4$ error on the input to~\eqnref{e:cz4to567and8to91011and12to131415flags}.  However, it is not enough to measure the $Z$ syndrome, or even to run full error correction, before applying the circuit~\eqnref{e:cz4to567and8to91011and12to131415flags}, because an $X_4$ or $Y_4$ fault could happen after the syndrome measurement completes and before~\eqnref{e:cz4to567and8to91011and12to131415flags}.  This problem is solved by the orange portion of~\eqnref{e:456712131415ftsyndrome}, which is meant to continue into~\eqnref{e:cz4to567and8to91011and12to131415flags}, replacing the first $\ket 0$ preparation and CNOT.  It gives qubit~$4$ temporary protection, so that an $X_4$ or $Y_4$ fault is caught by either the syndrome measurement or the orange $Z$ measurement, or both.  

While the above arguments give intuition for the construction, they leave out the details.  Let us now present the full fault-tolerant construction.  

1. Start by applying~\eqnref{e:456712131415ftsyndrome} to extract the syndrome for $Z_{\{4,5,6,7,12,13,14,15\}}$.  If the $Z$ or $X$ measurement is nontrivial, then decouple the orange qubit with another CNOT, apply error correction, and finish by applying unprotected $\CZ$ gates $4$ to $\{5,6,7\}$, $8$ to $\{9,10,11\}$ and $12$ to $\{13,14,15\}$.  (This is safe because one fault has already been detected.)  

2. Next, if the $Z$ and $X$ measurements were trivial, apply the top third of circuit~\eqnref{e:cz4to567and8to91011and12to131415flags}, where the orange qubit wire continues from~\eqnref{e:456712131415ftsyndrome}, to implement protected $\CZ$ gates $4$ to $\{5, 6, 7\}$.  If any measurements are nontrivial, then finish by applying unprotected $\CZ$ gates $8$ to $\{9,10,11\}$ and $12$ to $\{13,14,15\}$, then error correction.  We have argued already that this is fault tolerant; the extended orange ``flag" is enough to catch $X_4$ or $Y_4$ faults between~\eqnref{e:456712131415ftsyndrome} and~\eqnref{e:cz4to567and8to91011and12to131415flags}.  

3. If the measurements so far were trivial, then apply a circuit analogous to~\eqnref{e:456712131415ftsyndrome} to extract the $Z_{\{8,9,10,11,12,13,14,15\}}$ syndrome.  (Note that this is still a stabilizer, even though the $\CZ$ gates $4$ to $\{5, 6, 7\}$ have changed the code.)  If the $Z$ syndrome or $X$ measurement is nontrivial, then apply error correction---a simple error-correction procedure is to apply $\CZ$ gates $4$ to $\{5, 6, 7\}$ to move back to the $\llbracket 15, 7, 3 \rrbracket$ code and correct there---before finishing with $\CZ$ gates $8$ to $\{9, 10, 11\}$ and $12$ to $\{13, 14, 15\}$.  If the $Z$ and $X$ measurements were trivial, then apply the middle portion of~\eqnref{e:cz4to567and8to91011and12to131415flags}, where the orange qubit wire extends from qubit~$8$, to implement protected $\CZ$ gates $8$ to $\{9, 10, 11\}$.  If any measurements are nontrivial, then finish by applying unprotected $\CZ$ gates from $12$ to $\{13,14,15\}$, then error correction.  

4. If the measurements so far were trivial, then extract the $Z_{\{8,9,10,11,12,13,14,15\}}$ syndrome using~\eqnref{e:456712131415ftsyndrome} except with the data qubits in order $12, 13, 14, 15, 8, 9, 10, 11$ top to bottom (so that the orange flag attaches to qubit~$12$).  
If the $Z$ or $X$ measurement is nontrivial, then with at most one fault whatever error there is on the data can be corrected.  (The easiest way is to move forward to the $\llbracket 15, 7, 3 \rrbracket$ code using $\CZ$ gates $12$ to $\{13, 14, 15\}$ and correct there.  Note that these $\CZ$ gates turn the weight-one errors $X_{12}, X_{13}, X_{14}, X_{15}$ into $X_{12} Z_{\{13,14,15\}}, Z_{12} X_{13}, Z_{12} X_{14}, Z_{12} X_{15}$, respectively, but these can still be corrected; e.g., if in $X$ error correction you detect $X_{12}$, apply the correction $X_{12} Z_{\{13, 14, 15\}}$.)  
If the $Z$ and $X$ measurements were trivial, then apply the bottom portion of~\eqnref{e:cz4to567and8to91011and12to131415flags} to implement protected $\CZ$ gates from $12$ to $\{13, 14, 15\}$, and correct errors based on the measurement results.  

Observe that this procedure requires two ancilla qubits.

\subsection{Four-qubit fault-tolerant $\CCZ$ for universality} \label{s:1573ccz}

In order to realize a universal set of operations on the seven encoded qubits, we give a four-ancilla-qubit fault-tolerant implementation for an encoded $\CCZ$ gate.  The idea is to start with a circuit of round-robin $\CCZ$ gates to implement the encoded $\CCZ$ non-fault-tolerantly (\claimref{t:logicalccz}), then replace each $\CCZ$ with the gadget of Eq.~\eqnref{e:cczgadget} to catch correlated errors.  Finally, we intersperse $X$ error correction procedures to catch $X$ faults before they can spread, much like the pieceable fault-tolerance constructions of~\cite{YoderTakagiChuang16pieceableft} except on a single code block.  
(A similar approach can also be used to implement encoded $\CZ$ gates.)  

\smallskip

By \claimref{t:logicalccz}, an encoded $\CCZ$ gate can be implemented by round-robin $\CCZ$ gates on qubits $\{1, 4, 5\} \times \{1, 6, 7\} \times \{1, 8, 9\}$ (one $Z$, six $\CZ$ and $21$ $\CCZ$ gates): 
\begin{equation} \label{e:1573roundrobinccz1456789}
\raisebox{-1.3cm}{\includegraphics[scale=.769]{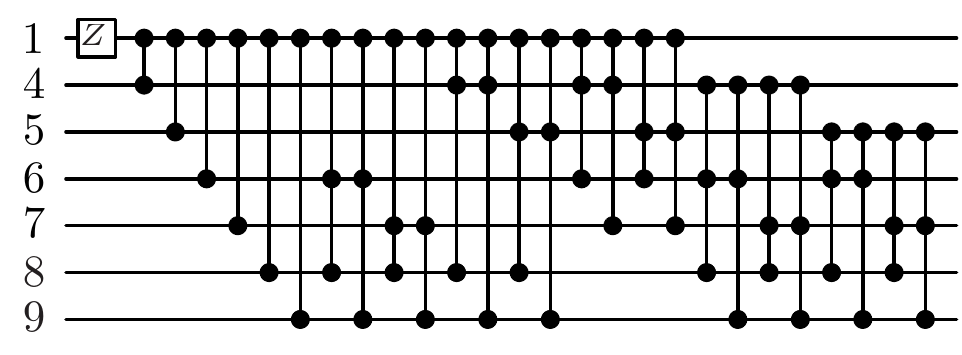}}
\end{equation}

To make this circuit fault tolerant, first replace each $\CZ$ gate with the gadget from \figref{f:czflags}, and replace each $\CCZ$ gate with the gadget from Eq.~\eqnref{e:cczgadget}.  After each gadget, apply $X$ error correction, and at the end apply both $X$ and $Z$ error correction.  (As in~\cite{YoderTakagiChuang16pieceableft}, it might be possible to put multiple $\CCZ$ gadgets before each $X$ error correction, but we have not tried to optimize this.)  Observe that $X$ error correction can be implemented even partially through the round-robin circuit because the code's $Z$ stabilizers are preserved by $\CCZ$ gates.  

There are two cases to consider to demonstrate fault tolerance: either a gadget is ``triggered" with a nontrivial, $1$ or $-$, measurement outcome, or no gadgets are triggered.  

\renewcommand{\theparagraph}{\arabic{paragraph}}

\paragraph{A gadget is triggered.}

If a gadget is triggered, then any Pauli errors can be present on its output data qubits.  It is straightforward to check mechanically that for each $\CZ$ gate in~\eqnref{e:1573roundrobinccz1456789}, all four possible $X$ errors, $II$, $IX$, $XI$ and $XX$, have distinct $Z$ syndromes, and so can be corrected immediately in the subsequent $X$ error correction, before the errors can spread.  By symmetry, the four possible $Z$ errors have distinct syndromes.  These errors commute through~\eqnref{e:1573roundrobinccz1456789} and are fixed by the final $Z$ error correction.  

Similar considerations hold for each $\CCZ$ gate: the possible $X$ and $Z$ error components have distinct syndromes, so an error's $X$ component can be corrected immediately and the $Z$ component corrected at the end.

\paragraph{No gadgets are triggered.}

If there is a single failure in a $\CZ$ or $\CCZ$ gadget, but the gadget is not triggered, then the error leaving the gadget is a linear combination of the same Paulis that could result from a one-qubit $X$, $Y$ or $Z$ fault before or after the gadget.  

If the error has no $X$ component, then as a weight-one $Z$ error it commutes to the end of~\eqnref{e:1573roundrobinccz1456789}, at which point $Z$ error correction fixes it.  

If the error has $X$ component of weight one, then the $Z$ component can be a permutation of any of $III, IIZ, IZZ, ZZZ$ on the three involved qubits (or of $II, IZ, ZZ$ for a $\CZ$ gadget).  As we have already argued, these $Z$ errors have distinct $X$ syndromes.  The $X$ error correction immediately following the gadget will catch and correct the error's $X$ component, keeping it from spreading.  The final $Z$ error correction, alerted to the $X$ failure, will correct the error's $Z$ component.

\section{Conclusion}

Space-saving techniques for fault-tolerant quantum computation should be useful both for large-scale quantum computers and for nearer-term fault-tolerance experiments.  Our techniques can likely be optimized further, and adapted to experimental model systems---but it might also be useful to relax the space optimization and allow a few more qubits.  The techniques can also likely be applied to other codes, especially distance-three CSS codes.  The round-robin $\CZ$ and $\CCZ$ constructions apply to some non-CSS codes, such as the $\llbracket 8,3,3 \rrbracket$ code, but then are more difficult to make fault tolerant.

\medskip

We thank Ted Yoder for helpful comments.  Research supported by NSF grant CCF-1254119 and ARO grant W911NF-12-1-0541.

\bibliography{q}

\end{document}